\documentclass[hidelinks,reqno,11pt]{article}
\usepackage{fullpage,amsmath,amsthm,amssymb,enumerate,enumitem,mathtools,float,accents,chngcntr,hyperref,url,graphicx,bm,algorithm, algorithmic}

\numberwithin{equation}{section}

\newtheorem{theorem}{Theorem}[section]
\newtheorem*{theorem*}{Theorem}
\newtheorem{lemma}[theorem]{Lemma}
\newtheorem*{lemma*}{Lemma}
\newtheorem{proposition}[theorem]{Proposition}
\newtheorem{corollary}[theorem]{Corollary}

\theoremstyle{definition}
\newtheorem{definition}[theorem]{Definition}
\newtheorem{remark}[theorem]{Remark}

\newtheorem{assumption}[theorem]{Assumption}

\DeclarePairedDelimiter{\abs}{\lvert}{\rvert}
\DeclarePairedDelimiter{\norm}{\lVert}{\rVert}
\DeclarePairedDelimiter{\paren}{(}{)}
\DeclarePairedDelimiter{\braces}{\lbrace}{\rbrace}
\DeclarePairedDelimiter{\inprod}{\langle}{\rangle}
\DeclarePairedDelimiter{\sqbracket}{[}{]}

\DeclarePairedDelimiter{\tripnorm}{\lvert\kern-0.25ex\lvert\kern-0.25ex\lvert}{\rvert\kern-0.25ex\rvert\kern-0.25ex\rvert}

\renewcommand{\P}{\mathbb{P}}
\newcommand{\Var}{\text{Var}}
\newcommand{\Cov}{\text{Cov}}
\newcommand{\E}{\mathbb{E}}
\newcommand{\R}{\mathbb{R}}
\newcommand{\C}{\mathbb{C}}
\newcommand{\Z}{\mathbb{Z}}

\newcommand{\bolda}{\textbf{a}}
\newcommand{\boldA}{\textbf{A}}
\newcommand{\boldB}{\textbf{B}}
\newcommand{\boldx}{\textbf{x}}
\newcommand{\boldy}{\textbf{y}}

\newcommand{\boldr}{\textbf{r}}
\newcommand{\boldv}{\textbf{v}}
\newcommand{\boldSigma}{\boldsymbol{\Sigma}}
\newcommand{\boldP}{\textbf{P}}
\newcommand{\boldQ}{\textbf{Q}}
\newcommand{\boldI}{\textbf{I}}
\newcommand{\boldX}{\textbf{X}}
\newcommand{\boldH}{\textbf{H}}

\newcommand{\cross}{\times}

\newlength{\dhatheight}

\newcommand{\Tr}{\textnormal{Tr}}

\newcommand{\diam}{\textnormal{diam}}

\title{Sparse Phase Retrieval via Sparse PCA despite Model Misspecification: A Simplified and Extended Analysis}

\author{Yan Shuo Tan \footnote{Department of Mathematics, University of Michigan, \href{mailto:yanshuo@umich.edu}{yanshuo@umich.edu}.}}

\begin{document}
	
\maketitle

\begin{abstract}
We consider the problem of high-dimensional misspecified phase retrieval. This is where we have an $s$-sparse signal vector $\boldx_*$ in $\R^n$, which we wish to recover using sampling vectors $\bolda_1,\ldots,\bolda_m$, and measurements $y_1,\ldots,y_m$, which are related by the equation $f(\inprod{\bolda_i,\boldx_*}) = y_i$. Here, $f$ is an unknown link function satisfying a positive correlation with the quadratic function. This problem was recently analyzed in \cite{Wang2016a}, which provided recovery guarantees for a two-stage algorithm with sample complexity $m = O(s^2\log n)$. In this paper, we show that the first stage of their algorithm suffices for signal recovery with the same sample complexity, and extend the analysis to non-Gaussian measurements. Furthermore, we show how the algorithm can be generalized to recover a signal vector $\boldx_*$ efficiently given geometric prior information other than sparsity.
\end{abstract}

\section{Introduction}

\subsection{Phase retrieval and sparse phase retrieval}

The phase retrieval problem is that of solving a system of quadratic equations
\begin{equation}	\label{eq: phase retrieval}
\abs{\inprod{\bolda_i,\boldx}^2} = y_i, \quad\quad i = 1,2,\ldots,m
\end{equation}
where $\bolda_i \in \R^n$ (or $\C^n$) are known sampling vectors,
$y_i > 0$ are observed measurements, and $\boldx \in \R^n$ (or $\C^n$) is the decision variable. Over the last half a decade, there has been great interest in constructing and analyzing algorithms with provable guarantees, with much success in the case where the sampling vectors are independently drawn from either a real or complex standard Gaussian distribution. The notable approaches include lifting and convex relaxation (\textit{PhaseLift}) \cite{Cand??s2013}, convex relaxation in the natural parameter space (\textit{PhaseMax}) \cite{Goldstein2016,Bahmani2016a,Hand2016a}, gradient descent \cite{Candes2015,Chen2015,Zhang2016,Wang2017}, stochastic gradient descent \cite{Tan2017,Jeong2017}, and prox-linear methods \cite{Duchi2017}. The best of these methods have been proved to accurately recover the true underlying signal $\boldx_*$ so long as the number of measurements $m$ is proportional to the signal dimension $n$.

Naturally, researchers have tried to replicate this success in the high-dimensional regime. In this setting, it is assumed that the true signal $\boldx_*$ is $s$-sparse, and one would like to estimate $x_*$ accurately with much fewer measurements than the ambient dimension, in analogy with what is possible for sparse linear regression. Work in this direction has mostly comprised straightforward adaptations of algorithms for unconstrained phase retrieval: Both \textit{PhaseLift} and \textit{PhaseMax} have been be adapted by adding $l_1$ regularizers to their respective objective functions \cite{Ohlsson2012,Hand2016}. Meanwhile, the gradient descent schemes \textit{Truncated Wirtinger Flow} and \textit{Truncated Amplitude Flow} have been adapted to alternate gradient steps with either soft- or hard-thresholding \cite{TonyCai2016,Wang2016,Soltanolkotabi2017}. These methods have been mostly shown to accurately recover $\boldx_*$ with sample complexity $m = \textnormal{O}^*(s^2)$.

\subsection{Single index models and model agnostic recovery}

Phase retrieval is an example of a single index model. In this more general setting, the measurements and the sampling vectors are related by the formula
\begin{equation}	\label{eq: single index model}
f(\inprod{\bolda_i,\boldx_*}) = y_i, \quad\quad i = 1,2,\ldots,m
\end{equation}
where $f$ is a (possibly random) link function. Such models have been studied for some time in the statistics community (see \cite{Kuchibhotla2007} and the references therein). Classically, it is assumed that the link function $f$ is unknown to the observer, and it is of interest to estimate both $\boldx_*$ (the index parameter) and $f$. Standard theoretical results in this body of work include asymptotic minimax rates of various estimators.

In this paper, we take a slightly different approach to the problem. We continue to assume that $f$ is unknown, but now treat $\boldx_*$ as the only parameter of interest. On the other hand, we are interested in algorithms that are provably efficient from both a statistical as well as a \textit{computational} point of view. Furthermore, we want our algorithms to be able to exploit a sparsity prior and thus work in the high-dimensional regime. The motivation for such an approach comes from the observation that real data almost never obeys a precise algebraic relationship. In other words, the neat relationships we postulate, such as \eqref{eq: phase retrieval}, are often \textit{misspecified}.

Recently, Plan and Vershynin \cite{Plan2016a} made significant progress on this problem in the setting of misspecified linear regression. They showed that if $\Cov(g,f(g)) \neq 0$, then the standard \textit{Lasso} algorithm for sparse linear regression is able to estimate $\boldx_*$ accurately up to scaling, and with a sample complexity of $\textnormal{O}(s\log n)$, the same order as that in the case of no model misspecification. Here, $g \sim \textnormal{N}(0,1)$ is a standard Gaussian random variable.

In the misspecified phase retrieval (MPR) setting, the first work was done by Neykov, Wang and Liu \cite{Wang2016a}. They proposed a two stage algorithm that works as follows. First, they form the reweighted sample covariance matrix
\begin{equation} \label{eq: hatSigma definition}
\hat{\boldSigma} := \frac{1}{m}\sum_{i=1}^m y_i \paren{\bolda_i\bolda_i^T-\boldI_n},
\end{equation}
and apply the standard SDP relaxation of \textit{Sparse PCA} to this matrix. Next, they use the leading eigenvector of the solution to formulate a \textit{Lasso}-type program. The solution to this program is their final estimate. The assumptions they make are that

\begin{equation} \tag{$A_{f,g}$}	\label{ass: Gaussian assumption}
\mu = \mu(g,f) := \Cov(g^2,f(g)) > 0,\quad \norm{f(g)}_{\psi_1} \leq C,
\end{equation}
under which, they were able to prove that the algorithm recovers $x_*$ accurately when given $m = \textnormal{O}^*(s^2)$ independent standard Gaussian sampling vectors. Again, this is the same order as the guarantees for sparse phase retrieval in the case of no misspecification.

\subsection{Contributions}

The contribution of this short paper is twofold. First, we prove that \textit{Sparse PCA}, the first step of the algorithm proposed in \cite{Wang2016a}, suffices to recover the signal vector $\boldx_*$ accurately with the same sample complexity as the full two-step algorithm given in their paper. We provide a simplified and more flexible analysis that is adapted from \cite{Plan2013}. This analysis has the further benefit of generalizing to the case where the prior assumption on $\boldx_*$ is not that it is sparse, but that it lies in a geometric set $\mathcal{K}$.

Second, we provide a guarantee for \textit{Sparse PCA} to recover $\boldx_*$ accurately when the sampling vectors are not Gaussian, but are instead drawn from distributions with independent subgaussian entries. In particular, we show that the method works for Rademacher random variables. Although this is a realistic sampling model, to our knowledge, it has not been analyzed in any prior work on phase retrieval. This guarantee requires two conditions. Unsurprisingly, we require the link function $f$ to satisfy a correlation condition similar to \eqref{ass: Gaussian assumption}, but adapted to the given subgaussian distribution. Second, we require $\boldx_*$ to have entries of equal magnitude over its support. This second condition is relatively stringent, but can probably be relaxed in future work.

\subsection{Notation}

We shall use boldface letters to denote vectors and matrices. If $\boldA$ and $\boldB$ are real matrices of the same dimensions, we let $\inprod{\boldA,\boldB} := \Tr(\boldA^T\boldB)$ denote the standard inner product. $\norm{\boldA}$ and $\norm{\boldA}_F$ will denote the operator and Frobenius norm of $\boldA$ respectively, while $\norm{\boldA}_1$ and $\norm{\boldA}_\infty$ will denote the entrywise $l_1$ and $l_\infty$ norms respectively. For a random variable $X$ and $\alpha > 0$, we let $\norm{X}_{\psi_{\alpha}}$ denote its $\psi_{\alpha}$ norm.\footnote{This is a tail decay condition. For more information, see Appendix \ref{sec: psi_alpha random variables}.} If $\bolda$ is a random vector, then $\bolda_i$ refers to an $i$-th independent copy of $\bolda$, while $(\bolda)_i$ refers to the $i$-th coordinate of $\bolda$. $C$ and $c$ are used to denote constants that may change from line to line.

\section{Main results}

We shall work with the single index model \eqref{eq: single index model}. We assume that the sampling vectors $\bolda_1,\ldots,\bolda_m$ are independent copies of a random vector $\bolda$ satisfying the following distributional assumption:
\begin{assumption}[Sampling vector distribution] \label{ass: sampling distribution}
The coordinates of $\bolda$ are independent copies of a random variable $Z$ that is centered, symmetric, of unit variance, and with subgaussian norm $\norm{Z}_{\psi_2}$ bounded by an absolute constant $C$.
\end{assumption}

For convenience, we shall hide the dependence on $C$ in our results and in our analysis. We do not assume that we know the link function $f$. The algorithm we propose to estimate $\boldx_*$ is the following.

\begin{algorithm}[H]
	\caption{{\sc Sparse PCA for MPR}}
	\begin{algorithmic}[1]              
		\REQUIRE Measurements $y_1,\ldots,y_m$, sampling vectors $\bolda_1,\ldots,\bolda_m$, sparsity level $s$.
		\ENSURE An estimate $\hat{\boldx}$ for $\boldx_*$.
		\STATE Compute $\hat{\boldSigma}$ as defined in \eqref{eq: hatSigma definition}.
		\STATE Let $\hat{\boldX}$ be the solution to
		\begin{equation} \label{eq: Sparse PCA SDP}
		\max_{X \succeq 0} ~\inprod{\boldX,\hat{\boldSigma}} \quad \textnormal{subject to} \quad \Tr(\boldX) = 1, ~\norm{\boldX}_1 \leq s.
		\end{equation}
		\STATE Let $\hat{\boldx}$ be the leading eigenvector to $\hat{\boldX}$.
	\end{algorithmic}
	\label{alg: Sparse PCA algorithm}
\end{algorithm}

This program is precisely the SDP relaxation of Sparse PCA proposed by d'Aspremont et al. \cite{DAspremont2007} and later analyzed by Amini and Wainwright \cite{Amini2009} and Berthet and Rigollet \cite{Berthet2013}. These two papers analyzed the performance of the algorithm as applied to sparse principal component detection in the spiked covariance model. Since the matrix $\hat{\boldSigma}$ does not follow this model a priori, one requires further analysis to show that the algorithm succeeds.

In \cite{Wang2016a}, the authors propose using the Lagrangian version of this program as the first step of their algorithm. Their analysis (see Lemma C.1. therein) shows that when the sampling vectors follow a Gaussian distribution, one can obtain a constant error approximation to $\boldx_*$ using $O(s^2\log n)$ samples. Using our methods, we prove a stronger version of this guarantee.

\begin{theorem}[Sparse recovery for Gaussian measurements] \label{thm: sparse recovery for Gaussian}
	Suppose $\bolda$ is a standard Gaussian in $\R^n$, and suppose Assumption \eqref{ass: Gaussian assumption} holds. Then there is an absolute constant $C$ such that for any $s$-sparse, unit norm signal $\boldx_*$ and any $\epsilon,\delta > 0$, the output $\hat{\boldx}$ to Algorithm \ref{alg: Sparse PCA algorithm} satisfies $\norm{\hat{\boldx}-\boldx_*}_2  \leq \epsilon$ with probability at least $1-\delta$ so long as the sample size $m$ satisfies
	\[
	m \geq C\max\braces*{\frac{s^2 \paren*{\log (n/\delta) + \log^4 (s/\delta)} }{\mu(f,g)^2\epsilon^4},\frac{s}{\delta},\frac{\log(n/\delta)}{\log^2 m}}.
	\]
\end{theorem}

Although this result is not entirely novel, we prove it in a different way compared to \cite{Wang2016a}. This method is simple and amenable to generalization to the situation where the sampling vectors are not Gaussian. In the non-Gaussian case, we first fix the the sparsity parameter $s$. Let $\bar{Z}_s := \frac{1}{s}\sum_{i=1}^s Z_i$ and $\boldr_{s,Z} := (Z_1 -\bar{Z}_s,\ldots,Z_s-\bar{Z}_s)$ denote the mean of $s$ independent copies of $Z$ and the vector of residuals respectively. With this notation, we make the assumption:
\begin{gather} \tag{$A_{f,Z,s}$} \label{ass: non-Gaussian assumption}
\mu(f,Z,s) := \Cov((\sqrt{s}\bar{Z}_s)^2,f(\sqrt{s}\bar{Z}_s)) > 0, \\
\sigma(f,Z,s) :=\Cov(\norm{\boldr_{s,Z}}_2^2,f(\sqrt{s}\bar{Z}_s)) \leq 0, \quad \norm{f(\sqrt{s}\bar{Z}_s)}_{\psi_1} \leq C.\nonumber
\end{gather}

Furthermore, we say that a unit norm signal vector $\boldx_*$ is \textit{admissible} if it has entries of equal magnitude across its support. In other words, there is a index set $I \subset [n]$ of cardinality $|I| \leq s$, such that
\[
(\boldx_*)_j = \begin{cases}
\pm \frac{1}{\sqrt{|I|}} & j \in I \\
0 & \text{otherwise}.
\end{cases}
\]

Using this definition, we have the following analogue of Theorem \ref{thm: sparse recovery for Gaussian}.

\begin{theorem}[Sparse recovery for non-Gaussian measurements] \label{thm: sparse recovery for non-Gaussians}
	There is an absolute constant $C$ such that the following holds. Fix a sparsity parameter $s$, suppose $\boldx_*$ is admissible and suppose Assumption \eqref{ass: non-Gaussian assumption} holds. Then for any $\epsilon,\delta > 0$, the output $\hat{\boldx}$ to Algorithm \ref{alg: Sparse PCA algorithm} satisfies $\norm{\hat{\boldx}-\boldx_*}_2  \leq \epsilon$ with probability at least $1-\delta$ so long as the sample size $m$ satisfies
	\begin{equation} \label{eq: sample complexity for non-Gaussian}
	m \geq \frac{Cs^2 \paren*{\log (n/\delta) + \log^4 (s/\delta)} }{\mu(f,Z,s)^2\epsilon^4}+\frac{Cs}{\delta} +\frac{C\log(n/\delta)}{\log^2 m}.
	\end{equation}
\end{theorem}

Note that when $Z$ is standard Gaussian, Assumption \eqref{ass: non-Gaussian assumption} reduces to Assumption \eqref{ass: Gaussian assumption}. To see this, observe that for any $s$, $\sqrt{s}\bar{Z}$ is a standard Gaussian random variable, while $\sigma(f,g,s) = 0$ by the independence property of orthogonal Gaussian marginals. This fact points to the assumption being the right generalization of \eqref{ass: Gaussian assumption}.

Furthermore, it is intuitive that the second condition should hold whenever $Z$ has a reasonable distribution and when $\mu(f,Z,s) > 0$: if $f$ is \textit{positively} correlated with the magnitude of $\bar{Z}$, then it should be \textit{negatively} correlated with the norm of the residual vector. Indeed, this can be shown to be true whenever $Z$ is a Rademacher random variable. We thus have a simpler result for Rademacher random variables:

\begin{corollary}[Sparse recovery for Rademacher measurements] \label{thm: sparse recovery for Rademachers}
	There is an absolute constant $C$ such that the following holds. Fix a sparsity parameter $s$, suppose $\boldx_*$ is admissible, let $Z$ denote a Rademacher random variable. Suppose $\mu(f,Z,s) > 0$ and $\norm{f(\sqrt{s}\bar{Z}_s)}_{\psi_1} \leq C$. Then for any $\epsilon,\delta > 0$, the output $\hat{\boldx}$ to Algorithm \ref{alg: Sparse PCA algorithm} satisfies $\norm{\hat{\boldx}-\boldx_*}_2  \leq \epsilon$ with probability at least $1-\delta$ so long as the sample size $m$ satisfies \eqref{eq: sample complexity for non-Gaussian}.
\end{corollary}

In the Gaussian setting, the recovery guarantee continues to hold even if we relax our constraint on $\boldx_*$ slightly and instead assume that $\norm{\boldx_*}_1 \leq \sqrt{s}$. This condition is \textit{geometric}: it can equivalently expressed as $\boldx_* \in \mathcal{K}$, where $\mathcal{K} = \sqrt{s}B_1^n$ is the $l_1$ norm ball. It is thus an interesting theoretical question to ask whether one can construct efficient algorithms for estimating $\boldx_*$ that exploit prior knowledge that $\boldx_* \in \mathcal{K}$ for a \textit{general} convex set $\mathcal{K}$.

There has been some work on proving \textit{statistical} efficiency guarantees for various algorithms. In the misspecified linear regression setting, Plan and Vershynin showed that the \textit{$\mathcal{K}$-Lasso} succeeds whenever the number of measurements $m$ is of the order $w(\mathcal{K})^2$, where $w(\mathcal{K})$ denotes the Gaussian width of $\mathcal{K}$ \cite{Plan2016a}. In the phase retrieval setting, Soltanolkotabi showed that Projected Amplitude Flow also succeeds whenever $m \gtrsim w(\mathcal{K})^2$. On the other hand, it is hard to remark on the \textit{computational} efficiency of these methods, because this depends on the properties of the set $\mathcal{K}$.

The final main result of this paper is a guarantee of a similar spirit.

\begin{theorem}[Recovery using general geometric constraints] \label{thm: general recovery}
	Suppose $\boldx_*\boldx_*^T \in \mathcal{K}$, where $\mathcal{K}$ is a convex subset of the space of unit trace PSD matrices in $\R^{n \cross n}$. Suppose $\bolda$ is a standard Gaussian in $\R^n$, and suppose Assumption \ref{ass: Gaussian assumption} holds. Then for any $\epsilon,\delta > 0$, the output $\hat{\boldx}$ to Algorithm \ref{alg: K-PCA algorithm} satisfies $\norm{\hat{\boldx}-\boldx_*}_2  \leq \epsilon$ with probability at least $1-\delta$ so long as the sample size $m$ satisfies
	\begin{equation} \label{eq: sample complexity for general constraints}
	m \geq \frac{C\paren*{w(\mathcal{K})^2 + \log^4(1/\delta) + \log m\paren{\gamma_1(\mathcal{K},\norm{\cdot}) + \log(1/\delta)}}}{\mu(f,g)^2\epsilon^4} + \frac{C}{\delta}.
	\end{equation}
	Here, $w(\mathcal{\mathcal{K}})$ and $\gamma_1(\mathcal{K},\norm{\cdot})$ respectively denote the Gaussian width of $\mathcal{K}$ and its $\gamma_1$-functional with respect to the operator norm, while $C$ is a universal constant.
\end{theorem}

\subsection*{Organization of paper and outline of proof strategy}

We prove Theorem \ref{thm: sparse recovery for Gaussian} and Theorem \ref{thm: sparse recovery for non-Gaussians} in Section \ref{sec: Proof of sparse recovery results}. The strategy we take comprises two steps. First, we compute the expected objective function used in Algorithm \ref{alg: Sparse PCA algorithm}, and show that it has sufficient curvature on the feasible set around the ground truth matrix, $\boldx_*\boldx_*^T$. This shows that feasible solutions having large expected objective value must also be close to $\boldx_*\boldx_*^T$. This computation is done in Section \ref{sec: cost function in expectation}.

Next, we argue that the empirical objective function is uniformly close to the expected objective function with high probability, so that a solution to the SDP program actually has large expected objective value. This is proved in Section \ref{sec: cost function concentration}. Finally, we use the same strategy for Theorem \ref{thm: general recovery}, but replace the objective function concentration analysis with a more sophisticated chaining argument. Due to its more technical nature, we defer the details to Appendix \ref{sec: general geometric signal constraints}.

\section{Proof of results for sparse recovery} \label{sec: Proof of sparse recovery results}

\begin{proof}[Proof of Theorem \ref{thm: sparse recovery for Gaussian}]
	Let $\boldX$ be the solution to Algorithm \ref{alg: Sparse PCA algorithm}. Since $\boldx_*\boldx_*^T$ is also feasible for the program, we have by optimality that
	\begin{equation} \label{eq: cost function expansion}
	0 \leq \inprod{\boldX-\boldx_*\boldx_*^T,\hat{\boldSigma}} = \inprod{\boldX-\boldx_*\boldx_*^T,\boldSigma} + \inprod{\boldX-\boldx_*\boldx_*^T,\hat{\boldSigma}-\boldSigma}.
	\end{equation}
	Using Lemma \ref{lem: cost function curvature}, the first term satisfies the bound
	\[
	\inprod{\boldX-\boldx_*\boldx_*^T,\boldSigma} \leq - \frac{\mu(f,g)}{2}\norm{\boldx_*\boldx_*^T - \boldX}_F^2.
	\]
	For the second term, we use H{\"o}lder's inequality to write
	\[
	\inprod{\boldX-\boldx_*\boldx_*^T,\hat{\boldSigma}-\boldSigma} \leq \norm*{\boldX-\boldx_*\boldx_*^T}_1 \norm{\hat{\boldSigma}-\boldSigma}_\infty.
	\]
	
	Next, we have by assumption that
	\[
	\norm*{\boldx_*\boldx_*^T}_1 = \sum_{i,j=1}^m \abs*{(\boldx_*)_i(\boldx_*)_j} = \norm{\boldx_*}_1^2 \leq s.
	\]
	Meanwhile, by construction, we also know that $\norm*{\boldX}_1 \leq s$. Rearranging \eqref{eq: cost function expansion}, we therefore get
	\[
	\frac{\mu(f,g)}{2}\norm{\boldx_*\boldx_*^T - \boldX}_F^2 \leq 2s \norm{\hat{\boldSigma}-\boldSigma}_\infty.
	\]
	Using Proposition \ref{prop: sample covariance matrix concentration} to bound the right hand side, we get
	\[
	\norm{\boldx_*\boldx_*^T - \boldX}_F^2 \leq \frac{Cs \paren*{\sqrt{\log (n/\delta)} + \log^2 (s/\delta)} }{\mu(f,g)\sqrt{m}}.
	\]
	
	If $\hat{\boldx}$ denotes the leading eigenvector of $\boldX$, we use Davis-Kahan's eigenvector perturbation theorem \cite{Davis1970a} to conclude that $\norm{\hat{\boldx}-\boldx_*}_2^2$ satisfies the same bound. Finally, we plug in our assumption on $m$ to show that this bound is less than $\epsilon^2$.
\end{proof}

\begin{proof}[Proof of Theorem \ref{thm: sparse recovery for non-Gaussians}]
	Exactly the same as for the Theorem \ref{thm: sparse recovery for Gaussian}.
\end{proof}

\begin{proof}[Proof of Corollary \ref{thm: sparse recovery for Rademachers}]
	Observe that $\norm{\boldr_{s,Z}}_2^2 + (\sqrt{s}\bar{Z}_s)^2 = \norm{\bolda}_2^2 = s$. We have
	\begin{align*}
	\sigma(f,Z,s) & = \Cov(f(\sqrt{s}\bar{Z}_s),\norm{\boldr_{s,Z}}_2^2) \\
	& = \Cov(f(\sqrt{s}\bar{Z}_s),s - (\sqrt{s}\bar{Z}_s)^2) \\
	& = - \Cov(f(\sqrt{s}\bar{Z}_s),(\sqrt{s}\bar{Z}_s)^2) = \mu(f,Z,s).
	\end{align*}
	The corollary now follows from Theorem \ref{thm: sparse recovery for non-Gaussians}.
\end{proof}

\section{Objective function in expectation} \label{sec: cost function in expectation}

In this section, we compute expressions for the expected reweighted covariance matrix $\boldSigma = \E\hat{\boldSigma}$. Note that we may also write
\[
\boldSigma = \E y \paren{\bolda\bolda^T-\boldI_n}.
\]

\begin{lemma}[Expected covariance for Gaussian distribution] \label{lem: Expected covariance for Gaussians}
	Suppose $\bolda \sim \textnormal{N}(0,\boldI_n)$. Then for any $\boldx_* \in S^{n-1}$, we have
	\[
	\boldSigma = \mu(f,g) \boldx_*\boldx_*^T.
	\]
\end{lemma}

\begin{proof}
	Decompose $\bolda = \inprod{\bolda,\boldx_*}\boldx_* + \bolda^\perp$, where $\bolda^\perp$ is the projection of $\bolda$ to the orthogonal complement of $\boldx_*$. Using this, we write
	\begin{align*}
	\E \braces{y\bolda\bolda^T} & = \E \braces{f(\inprod{\bolda,\boldx_*}) (\inprod{\bolda,\boldx_*}\boldx_* + \bolda^\perp)(\inprod{\bolda,\boldx_*}\boldx_* + \bolda^\perp)^T} \\
	& = \E \braces{f(\inprod{\bolda,\boldx_*})\inprod{\bolda,\boldx_*}^2\boldx_*\boldx_*^T} + \E \braces{f(\inprod{\bolda,\boldx_*}) \bolda^\perp (\bolda^\perp)^T} \\	
	& \quad + \E \braces{f(\inprod{\bolda,\boldx_*})\inprod{\bolda,\boldx_*}\boldx_*(\bolda^\perp)^T} + \E \braces{f(\inprod{\bolda,\boldx_*})\inprod{\bolda,\boldx_*}\bolda^\perp\boldx_*^T}.
	\end{align*}
	Because $\bolda$ is a standard Gaussian, $\inprod{\bolda,\boldx_*}$ and $\bolda^\perp$ are independent. This means that the third and fourth terms in this sum are zero. Furthermore, the second term can be written as the product of two expectations $\E \braces{f(\inprod{\bolda,\boldx_*})}$ and $\E \braces{\bolda^\perp (\bolda^\perp)^T}$. We now use standard computations for Gaussians to continue writing
	\begin{align*}
	\E \braces{y(\bolda\bolda^T-\boldI_n)} & = \E \braces{f(g)g^2} x_*x_*^T +  \E \braces{f(g)}(\boldI_n - x_*x_*^T) + \E \braces{f(g)}\boldI_n \\
	& = \mu(f,g) \boldx_*\boldx_*^T.
	\end{align*}
	This completes the proof.
\end{proof}

\begin{lemma}[Expected covariance for non-Gaussian distributions] \label{lem: Expected covariance for non-Gaussians}
	Suppose $\bolda$ is a random vector in $\R^n$ that satisfies Assumption \ref{ass: sampling distribution}. Let $2 \leq s \leq n$ be an integer, and let $\boldx_*$ be an admissible signal vector. We have
	\[
	\boldSigma = \mu(f,Z,s) \boldx_*\boldx_*^T - \frac{\sigma(f,Z,s)}{s-1}(\boldP_I-\boldx_*\boldx_*^T).
	\]
\end{lemma}

\begin{proof}
	Let $\boldP_I$ and $\boldP_I^\perp$ denote the orthogonal projections to the coordinates in $I$ and $I^c$ respectively. Then $\boldP_I\bolda$ and $\boldP_I^\perp\bolda$ are independent. Using a similar calculation as in the previous lemma, we see that $\boldP_I\Sigma\boldP_I^\perp = \boldP_I^\perp\Sigma\boldP_I = \boldP_I^\perp\Sigma\boldP_I^\perp = 0$. We may hence assume WLOG that $s=n$ and $I = [n]$. By the symmetry of the distribution of $\bolda$, we may also assume that $\boldx_* = \frac{\textbf{1}}{\sqrt{n}}$, where $\textbf{1}$ is the all ones vector. 
	
	Next, notice that $\inprod{\bolda,\boldx_*}$ is invariant to permutation of the coordinate indices. Meanwhile, the \textit{distributions} of $\bolda^\perp$ and $\bolda$ are both symmetric with respect to such transformations. Let $\boldQ$ be a permutation matrix. Then
	\begin{align*}
	\boldQ\boldSigma\boldQ^T & = \E\braces{f(\inprod{\bolda,\boldx_*})\boldQ\bolda(\boldQ\bolda)^T} \\
	& = \E\braces{f(\inprod{\boldQ\bolda,\boldx_*})\boldQ\bolda(\boldQ\bolda)^T} \\
	& = \E\braces{f(\inprod{\bolda,\boldx_*})\bolda(\bolda)^T}.
	\end{align*}
	In other words, we have
	\begin{equation} \label{eq: permutation invariance}
	\boldQ\boldSigma\boldQ^T = \boldSigma.
	\end{equation}
	One can check that a matrix satisfying \eqref{eq: permutation invariance} for all permutation matrices $\boldQ$ must have the same value for all diagonal entries, and the same value for all off-diagonal entries. In other words, $\boldSigma$ must be of the form
	\begin{equation} \label{eq: Sigma decomposition}
	\boldSigma = \frac{\alpha}{n} \textbf{1}\textbf{1}^T + \beta (\boldI_n -\frac{1}{n}\textbf{1}\textbf{1}^T)
	\end{equation}
	for some values of $\alpha$ and $\beta$.
	
	Let us now compute the values of $\alpha$ and $\beta$ using the fact that $\boldx_* = \frac{\textbf{1}}{\sqrt{n}}$. We have
	\[
	\alpha = \boldx^T\boldSigma\boldx_* = \mu(f,Z,n).
	\]
	Next, we apply traces to \eqref{eq: Sigma decomposition} to get
	\[
	\alpha + (n-1)\beta = \Tr(\boldSigma) = \E\braces{f(\inprod{\bolda,\boldx_*})\Tr(\bolda\bolda^T - \boldI_n)}.
	\]
	Observe further that
	\begin{align*}
	\E\braces{f(\inprod{\bolda,\boldx_*})\Tr(\bolda\bolda^T - \boldI_n)} = \E\braces{f(\inprod{\bolda,\boldx_*})\paren{\norm{\bolda}_2^2 - n}}
	& = \sigma(f,Z,n) + \mu(f,Z,n).
	\end{align*}
	As such, we have $\beta = \frac{\sigma(f,Z,n)}{n-1}$ as we wanted.
\end{proof}

\begin{lemma}[Curvature of objective function] \label{lem: cost function curvature}
	Suppose the hypotheses of Lemma \ref{lem: Expected covariance for Gaussians} (respectively Lemma \ref{lem: Expected covariance for non-Gaussians}) hold. For any $\boldX \succeq 0$ such that $\Tr(\boldX) = 1$, we have
	\begin{equation}
	\inprod{\boldSigma,\boldx_*\boldx_*^T - \boldX} \geq \frac{\mu}{2}\norm{\boldx_*\boldx_*^T - \boldX}_F^2,
	\end{equation}
	where $\mu = \mu(f,g)$ (respectively $\mu = \mu(f,Z,s)$).
\end{lemma}

\begin{proof}
	We shall prove the case where the hypotheses of Lemma \ref{lem: Expected covariance for non-Gaussians} hold. The other case is similar and even easier. First, observe that
	\[
	\inprod{\boldSigma,\boldx_*\boldx_*^T} = \mu(f,Z,s)\inprod{\boldx_*\boldx_*^T,\boldx_*\boldx_*^T} = \mu(f,Z,s) \norm{\boldx_*}_2^4.
	\]
	We also have
	\begin{align*}
	\inprod{\boldSigma,\boldX} & = \mu(f,Z,s)\inprod{\boldx_*\boldx_*^T,\boldX} + \frac{\sigma(f,Z,s)}{s-1}\inprod{\boldP_I-\boldx_*\boldx_*^T,\boldX} \\
	& = \mu(f,Z,s)\inprod{\boldx_*\boldx_*^T,\boldX} + \frac{\sigma(f,Z,s)}{s-1} \Tr(\boldX\big|_{\R^I\cap\boldx_*^\perp}) \\
	& \leq \mu(f,Z,s)\inprod{\boldx_*\boldx_*^T,\boldX}.
	\end{align*}
	Here, the last inequality follows from the fact that $\boldX$ is positive semidefinite, which implies that any partial trace has to be non-negative.
	
	Now, the assumptions on $\boldX$ also imply that $\norm{\boldX}_F \leq 1 = \norm{\boldx_*}_2$. We can thus combine our calculations to get
	\begin{align*}
	\inprod{\boldSigma,\boldx_*\boldx_*^T - \boldX} & \geq \mu(f,Z,s) \norm{\boldx_*}_2^4 - \mu(f,Z,s)\inprod{\boldx_*\boldx_*^T,\boldX} \\
	& \geq \frac{\mu(f,Z,s)}{2} \paren*{\norm{\boldx_*\boldx_*^T}_F + \norm{\boldX}_F^2 - 2 \inprod{\boldx_*\boldx_*^T,\boldX}} \\
	& = \frac{\mu(f,Z,s)}{2}\norm{\boldx_*\boldx_*^T - \boldX}_F^2.
	\end{align*}
	This completes the proof.
\end{proof}

\section{Concentration of objective function} \label{sec: cost function concentration}

The goal of this section is to prove the following concentration theorem for the reweighted sample covariance matrix $\hat{\boldSigma}$.

\begin{proposition}[Concentration of sample matrix] \label{prop: sample covariance matrix concentration}
	There is a universal constant $C$ so that the following holds. Fix a sparsity parameter $s$, let $\bolda$ be defined using Assumption \ref{ass: sampling distribution}. Suppose Assumption \eqref{ass: non-Gaussian assumption} holds, and let $\boldx_*$ be a unit  norm vector. If $\bolda$ is non-Gaussian, further assume that $\boldx_*$ is admissible. Then for any $\delta > 0$, we have 
	\[
	\norm{\hat{\boldSigma} - \boldSigma}_\infty \leq \frac{C (\sqrt{\log (n/\delta)} + \log^2 (s/\delta)) }{\sqrt{m}}
	\]
	with probability at least $1-\delta$, provided $m \geq C\max\braces{s/\delta,\log(n/\delta)\log^2 m }$.
\end{proposition}

\begin{proof}
	Without loss of generality, assume that the support of $\boldx_*$ is contained in the first $s$ coordinates. Let $\boldP_s$ denote the projection to the first $s$ coordinates. We write
	\begin{equation} \label{eq: max entry decomposition}
	\norm{\hat{\boldSigma} - \boldSigma}_\infty = \max\braces*{\norm{\boldP_s(\hat{\boldSigma} - \boldSigma)\boldP_s}_\infty, \norm{\boldP_s(\hat{\boldSigma} - \boldSigma)\boldP_s^\perp}_\infty,\norm{\boldP_s^\perp(\hat{\boldSigma} - \boldSigma)\boldP_s^\perp}_\infty},
	\end{equation}
	and bound each of the terms on the right separately.
	
	For the first term, we shall use the fact that each entry is the mean of $m$ i.i.d. $\psi_{1/2}$ random variables (see Appendix \ref{sec: psi_alpha random variables}). This tail decay gives us a relatively strong large deviation inequality, which we can use together with a union bound. In more detail, let $1 \leq k,l \leq s$. Then
	\[
	(\hat{\boldSigma}-\boldSigma)_{kl} = \frac{1}{m} \sum_{i=1}^m \Big[(\bolda_i)_k(\bolda_i)_l f(\inprod{\bolda_i,\boldx_*}) - \E\braces{ (\bolda_i)_k(\bolda_i)_l f(\inprod{\bolda_i,\boldx_*})}\Big].
	\]
	We now use Proposition \ref{prop: psi_alpha norm of centered RV} followed by Proposition \ref{prop: submultiplicativity of psi norm} twice to get
	\begin{align*}
	\norm{(\bolda)_k(\bolda)_l f(\inprod{\bolda,\boldx_*}) - \E\braces{ (\bolda)_k(\bolda)_l f(\inprod{\bolda,\boldx_*})}}_{\psi_\alpha} & \lesssim \norm{(\bolda)_k(\bolda)_l f(\inprod{\bolda,\boldx_*})}_{\psi_{1/2}} \\
	& \lesssim 
	\norm{(\bolda)_k(\bolda)_l}_{\psi_1} \norm{f(\inprod{\bolda,\boldx_*})}_{\psi_1} \\
	& \lesssim \norm{(\bolda)_k}_{\psi_2}\norm{(\bolda)_l}_{\psi_2} \norm{f(\inprod{\bolda,\boldx_*})}_{\psi_1}.  
	\end{align*}
	Each of the terms in the product on the right hand side is bounded by an absolute constant by assumption. As such, the quantity on the left is also bounded by an absolute constant. We may thus use Proposition \ref{prop: Bernstein for psi_1/2} to see that
	\[
	\P\braces{ \abs{(\hat{\boldSigma}-\boldSigma)_{kl}} > t/\sqrt{m} } \leq 2\exp(-c\sqrt{t})
	\]
	for $t > 0$ large enough. Pick $t \sim \log^2(s/\delta)$. Then we can take a union bound over all $s^2$ choices of $k$ and $l$ to get
	\[
	\norm{\boldP_s(\hat{\boldSigma} - \boldSigma)\boldP_s}_\infty \lesssim \frac{\log^2(s/\delta)}{\sqrt{m}}
	\]
	with probability at least $1- \delta/4$.
	
	We next bound the other two quantities in \eqref{eq: max entry decomposition} via a conditioning argument similar to that in \cite{Wang2016a}. The key idea is to condition on the probability $1-\delta/4$ event over which the three statements in Lemma \ref{lem: helper lemma for matrix concentration} hold, and to observe that this event is \textit{independent} of the random variables $(\bolda_i)_k$ for $1 \leq i \leq m$, $s < k \leq n$. Hence, conditioning on the event does not alter the joint distribution of this set of random variables.
	
	We consider a typical entry in $\boldP_s(\hat{\boldSigma} - \boldSigma)\boldP_s^\perp$, which is of the form
	\begin{equation} \label{eq: off diagonal block entry}
	\frac{1}{m}\sum_{i=1}^m (\bolda_i)_k(\bolda_i)_l f(\inprod{\bolda_i,\boldx_*}), \quad 1 \leq k \leq s, ~s < l \leq n.
	\end{equation}
	Fixing all randomness apart from $(\bolda_i)_l$ for all indices $1 \leq i \leq m$, $s < l \leq n$, we can use Hoeffding's inequality (Proposition \ref{prop: Hoeffding}) to conclude that for each $l$, \eqref{eq: off diagonal block entry} is a subgaussian random variable with variance $\frac{1}{m^2}\sum_{i=1}^m (\bolda_i)_k^2 f(\inprod{\bolda_i,\boldx_*})^2$. By the second statement of Lemma \ref{lem: helper lemma for matrix concentration}, this is bounded by $C/m$, so that
	\begin{equation}
	\P\braces*{\abs*{\frac{1}{m}\sum_{i=1}^m (\bolda_i)_k(\bolda_i)_l f(\inprod{\bolda_i,\boldx_*})} > \frac{t}{\sqrt{m}}} \leq 2\exp(-ct^2)
	\end{equation}
	Choosing $t \sim \sqrt{\log(n/\delta)}$ and taking a union bound over $s < l \leq n$ gives
	\[
	\norm{\boldP_s(\hat{\boldSigma} - \boldSigma)\boldP_s^\perp}_\infty \lesssim \sqrt{\frac{\log(n/\delta)}{m}}
	\]
	with probability at least $1-\delta/4$.
	
	Finally, each entry of $\boldP_s^\perp(\hat{\boldSigma} - \boldSigma)\boldP_s^\perp$ is of the form
	\begin{equation} \label{eq: lower diagonal block entry}
	\frac{1}{m}\sum_{i=1}^m f(\inprod{\bolda_i,\boldx_*})\big[(\bolda_i)_k(\bolda_i)_l - \E\braces{(\bolda_i)_k(\bolda_i)_l}\big] , \quad s < k,l \leq n.
	\end{equation}
	We again fix all randomness apart from $(\bolda_i)_l$ for all indices $1 \leq i \leq m$, $s < l \leq n$. Observe that $(\bolda_i)_k(\bolda_i)_l - \E\braces{(\bolda_i)_k(\bolda_i)_l}$, $s < k,l \leq n$, are centered subexponential random variables. We may thus use Bernstein's inequality (Proposition \ref{prop: Bernstein}) together with the second and third statements of Lemma \ref{lem: helper lemma for matrix concentration} to obtain the tail bound:
	\begin{equation}
	\P\braces*{\abs*{\frac{1}{m}\sum_{i=1}^m f(\inprod{\bolda_i,\boldx_*})\big[(\bolda_i)_k(\bolda_i)_l - \E\braces{(\bolda_i)_k(\bolda_i)_l}\big]} > \frac{t}{\sqrt{m}}} \leq 2\exp\paren*{-c \min\braces*{t^2,\frac{t\sqrt{m}}{\log m}}}
	\end{equation}
	Once again, choosing $t \sim \sqrt{\log(n/\delta)}$ and taking a union bound over $s < k,l \leq n$ gives
	\[
	\norm{\boldP_s(\hat{\boldSigma} - \boldSigma)\boldP_s^\perp}_\infty \lesssim \sqrt{\frac{\log(n/\delta)}{m}},
	\]
	with probability at least $1-\delta/4$ provided that $m \gtrsim \log(n/\delta)\log^2 m$.
	\end{proof}

\begin{remark}
	When $\bolda$ is a standard Gaussian, \cite{Wang2016a} gave the bound
	\[
	\norm{\hat{\boldSigma} - \boldSigma}_\infty \leq \frac{C \sqrt{\log (n/\delta)} }{\sqrt{m}}
	\]
	with roughly the same tail probability. Hence, the only price to having more distributional generality is the additional $\log^2(s/\delta)$ term in the numerator.
\end{remark}

\begin{lemma} \label{lem: helper lemma for matrix concentration}
	Let the hypotheses of Proposition \ref{prop: sample covariance matrix concentration} hold. There is an absolute constant $C$ such that the following holds. Let $I$ denote the support of $\boldx_*$. Then for any $\delta > 0$, so long as $m \geq Cs/\delta$, the following three statements hold simultaneously with probability at least $1- \delta/4$.
	\begin{enumerate}
		\item $\displaystyle{\sum_{i=1}^m f(\inprod{\bolda_i,\boldx_*})^2 \leq Cm}$.
		\item $\displaystyle{\max_{k \in I}\sum_{i=1}^m (\bolda_i)_k^2 f(\inprod{\bolda_i,\boldx_*})^2 \leq Cm}$.
		\item $\displaystyle{\max_{1 \leq i \leq m}f(\inprod{\bolda_i,\boldx_*}) \leq C\log m }$.
	\end{enumerate}
\end{lemma}

\begin{proof}
	By Assumption \eqref{ass: non-Gaussian assumption}, we know that $\norm{f(\inprod{\bolda_i,\boldx_*}) }_{\psi_1}$ is bounded by an absolute constant. As such, Proposition \ref{prop: characterization of psi_alpha RVs} implies that both its second and fourth moments are also bounded. Furthermore, we have
	\[
	\Var(f(\inprod{\bolda_i,\boldx_*})^2) \leq \E\braces{f(\inprod{\bolda_i,\boldx_*})^4} \leq C,
	\]
	where $C$ is an absolute constant. Using Chebyshev's inequality together with the second moment bound, we thus get
	\begin{equation} \label{eq: Chebyshev bound 1}
	\P\braces*{ \sum_{i=1}^m f(\inprod{\bolda_i,\boldx_*})^2 \geq m(C + t)} \leq \frac{C}{mt^2}.
	\end{equation}
	We can use the same argument together with a union bound over $k \in I$ to get
	\begin{equation} \label{eq: Chebyshev bound 2}
	\P\braces*{ \max_{k \in I}\sum_{i=1}^m (\bolda_i)_k^2f(\inprod{\bolda_i,\boldx_*})^2 \geq m(C + t)} \leq \frac{Cs}{mt^2}.
	\end{equation} \label{eq: subexponential bound}
	Finally, we again use the union bound and the subexponential tail bound to get
	\begin{equation}
	\P\braces*{\max_{1\leq i \leq m}f(\inprod{\bolda_i,\boldx_*}) \geq t\log m} \leq 2m\exp(-ct\log m) = 2m^{1-ct}.
	\end{equation}
	
	Choose $t$ to be any fixed constant in \eqref{eq: Chebyshev bound 1} and \eqref{eq: Chebyshev bound 2}, and choose $t$ to be a constant larger than $2/c$ in \eqref{eq: subexponential bound}. Then each of these probability bounds is of the order $O(1/m)$, so that $m \gtrsim s/\delta$ suffices for all three statements to hold with probability at least $1-\delta/4$.
\end{proof}

\section{Conclusion and discussion}

In this paper, we have analyzed the problem of misspecified phase retrieval, and improved upon the work of Neykov et al. in \cite{Wang2016a}. In particular, we have shown that the first stage of their algorithm suffices for signal recovery with the same sample complexity, and extended the analysis to non-Gaussian measurements. Furthermore, we showed how the algorithm can be generalized to recover a signal vector $\boldx_*$ efficiently given geometric prior information other than sparsity.

Experts in compressed sensing may have observed that while the sample complexity for algorithms for misspecified linear regression scales linearly with the sparsity parameter, our sample complexity bounds for misspecified phase retrieval scale instead with the square of the parameter. In \cite{Wang2016a}, the authors showed numerical evidence that this discrepancy is due to the statistical inefficiency of the algorithm, and not merely a slackness in the mathematical analysis.

This $s^2$ scaling is also observed in all other efficient algorithms for sparse phase retrieval, and it is an open question whether there exist computationally efficient algorithms that can do better. The authors of \cite{Wang2016a} conjecture that the answer is in the negative. This is supported by results by Berthet and Rigollet, who show that computationally efficient algorithms for the related problem of detecting sparse principal components, using $O(s^{2-\epsilon})$ samples for any $\epsilon > 0$, will lead to computationally efficient algorithms for solving hard instances of the planted clique problem \cite{Berthet2013,Berthet2013c}. This is widely conjectured to be impossible.

It will also be interesting to investigate whether there is slackness in the sample complexity bound for signal recovery using general geometric constraints (Theorem \ref{thm: general recovery}). In particular, I do not know how to bound $\gamma_1(\mathcal{K},\norm{\cdot})$ where $\mathcal{K}$ is the set of unit trace PSD matrices $\boldX$ with $\norm{\boldX}_1 \leq s$. Hence, it is not yet clear whether Theorem \ref{thm: sparse recovery for Gaussian} can be derived from Theorem \ref{thm: general recovery}.

Finally, the literature on high-dimensional signal recovery from non-Gaussian measurements is still fairly limited. In this work, we have proved a recovery guarantee for \textit{admissible} signal vectors in the case of misspecified phase retrieval. Hopefully, this guarantee can be extended to larger classes of signal vectors in the near future.

\subsection*{Acknowledgements}

I would like to thank Roman Vershynin and Shahar Mendelson for helpful discussions. Part of this manuscript was written on a visit to UC Irvine supported by U.S. Air Force Grant FA9550-18-1-0031.

\nocite{*}
\bibliographystyle{acm}
\bibliography{Projects-sparse_PR}

\appendix

\section{Results on $\psi_{\alpha}$ random variables} \label{sec: psi_alpha random variables}

In this section, we collect some standard definitions and results on $\psi_{\alpha}$ random variables that help us give streamlined proofs for our main theorems. These are collated from \cite{Ledoux1991}, \cite{VanderVaart1996}, \cite{Vershynin}, and \cite{Ma2015}.

\begin{definition}[Orlicz norms]
	Let $\psi\colon \R_+ \to \R_+$ be a convex, increasing function with $\psi(0) = 0$. Define the \textit{Orlicz norm} of a random variable $X$ with respect to $\psi$ as
	\[
	\norm{X}_\alpha := \inf\braces*{\lambda > 0 ~\colon~ \E \braces{\psi(|X|/\lambda)} \leq 1}.
	\]
	Equipped with this norm, the space of random variables with finite norm forms a Banach space, called an \textit{Orlicz space}.
\end{definition}

We are especially interested in the Orcliz spaces corresponding to $\psi_{\alpha}$ for $\alpha > 0$. These are defined as follows. When $\alpha \geq 1$, we set $\psi_{\alpha}(x) := \exp(x^\alpha)-1$. When $0 < \alpha \leq 1$, this function is no longer convex, so we convexify it by fiat, setting $\psi_{\alpha}(x) := \exp(x^\alpha)-1$ for $x \geq x(\alpha)$ large enough, and taking $\psi_{\alpha}$ to be linear on $[0,x(\alpha)]$. If some random variable $X$ has a finite $\psi_{\alpha}$ norm $\norm{X}_{\psi_{\alpha}}$, we say that it is a \textit{$\psi_{\alpha}$ random variable}.

Readers will probably be familiar with $\psi_2$ and $\psi_1$ Orcliz spaces. These correspond to subgaussian and subexponential random variables respectively (see \cite{Vershynin} for more details). For these two classes of random variables, we have the well-known Hoeffding's and Bernstein's inequalities.

\begin{proposition}[Hoeffding's inequality] \label{prop: Hoeffding}
	Let $X_1,\ldots,X_m$ be independent, centered, subgaussian random variables. Then for every $t \geq 0$, we have
	\[
	\P\braces*{\abs*{\sum_{i=1}^m X_i}} \leq 2\exp\paren*{-\frac{ct^2}{\sum_{i=1}^m \norm{X_i}_{\psi_2}^2}},
	\]
	where $c > 0$ is an absolute constant.
\end{proposition}

\begin{proposition}[Bernstein's inequality] \label{prop: Bernstein}
	Let $X_1,\ldots,X_m$ be independent, centered, subexponential random variables. Then for every $t \geq 0$, we have
	\[
	\P\braces*{\abs*{\sum_{i=1}^m X_i}} \leq 2\exp\paren*{-c\min\braces*{\frac{t^2}{\sum_{i=1}^m \norm{X_i}_{\psi_1}^2},\frac{t}{\max_{i}\norm{X_i}_{\psi_1}}}},
	\]
	where $c > 0$ is an absolute constant.
\end{proposition}

In this paper however, it will be useful for us to consider Orlicz spaces in full generality. This is because we will need to work with $\psi_{1/2}$ random variables, for which many of the standard concentration inequalities do not hold. Nonetheless, we still have the following.

\begin{proposition}[Characterization of $\psi_{1/2}$ RVs] \label{prop: characterization of psi_alpha RVs}
	Let $X$ be a real-valued random variable. Then the following properties are equivalent. The parameters $C_i > 0$ appearing in these properties differ from each other by at most an absolute contant factor.
	\begin{enumerate}
		\item The tails of $X$ satisfy
		\[
		\P\braces*{\abs{X} \geq t} \leq 2\exp\paren*{-\sqrt{t/C_1}}.
		\]
		\item The moments of $X$ satisfy
		\[
		\norm{X}_p = \paren*{\E\abs*{X}^p}^{1/p} \leq C_2 p^2.
		\]
		\item The $\psi_{1/2}$ norm of $X$ satisfies
		\[
		\norm{X}_{\psi_{1/2}} \leq C_3.
		\]
	\end{enumerate}
\end{proposition}

\begin{proof}
	Same as in the case of $\psi_1$ and $\psi_2$. See \cite{Vershynin}.
\end{proof}

We have the following further properties.

\begin{proposition}[Products, Lemma 8.5 in \cite{Ma2015}] \label{prop: submultiplicativity of psi norm}
	Let $X$ and $Y$ be $\psi_{\alpha}$ random variables for some $\alpha > 0$. Then $XY$ is a $\psi_{\alpha/2}$ random variable with $\psi_{\alpha/2}$ norm satisfying
	\[
	\norm{XY}_{\psi_{\alpha/2}} \leq C_\alpha\norm{X}_{\psi_\alpha}\norm{Y}_{\psi_\alpha}.
	\]
	Here, $C_\alpha$ is an absolute constant depending only on $\alpha$.
\end{proposition}

\begin{proposition}[Centering] \label{prop: psi_alpha norm of centered RV}
	Let $X$ be a $\psi_{\alpha}$ random variable for some $\alpha > 0$. Then 
	\[
	\norm{X -\E X}_{\psi_{\alpha}} \leq 2\norm{X}_{\psi_\alpha}.
	\]
\end{proposition}

\begin{proof}
	We have $\norm{X -\E X}_{\psi_{\alpha}} \leq \norm{X}_{\psi_\alpha} + \norm{\E X}_{\psi_\alpha}$. Now check the definition of the norm to verify that $\norm{\E X}_{\psi_\alpha} \leq \norm{X}_{\psi_\alpha}$.
\end{proof}

\begin{proposition}[Sums, Theorem 6.21 in \cite{Ledoux1991}] \label{prop: bound on psi_alpha norm of sum}
	Let $0 < \alpha \leq 1$, and let $X_1,\ldots,X_m$ be a sequence of independent, centered $\psi_{\alpha}$ random variables. Then
	\[
	\norm*{\sum_{i=1}^m X_i}_{\psi_\alpha} \leq C_\alpha \paren*{\E\abs*{\sum_{i=1}^m X_i} + \norm*{\max_{1 \leq i \leq m}\abs{X_i}}_{\psi_{\alpha}} }.
	\]
	Here, $C_\alpha$ is an absolute constant depending only on $\alpha$.
\end{proposition}

\begin{proposition}[Maxima, Lemma 2.2.2 in \cite{VanderVaart1996}] \label{prop: bound on psi_alpha norm of max}
	Let $0 < \alpha \leq 1$, and let $X_1,\ldots,X_m$ be independent, centered $\psi_{\alpha}$ random variables. Then
	\[
	\norm*{\max_{1 \leq i \leq m}\abs{X_i}}_{\psi_{\alpha}} \leq C_\alpha \psi_{\alpha}^{-1}(m) \max_{1 \leq i \leq m} \norm{X_i}_{\psi_{\alpha}}.
	\]
	Here, $C_\alpha$ is an absolute constant depending only on $\alpha$.
\end{proposition}

\begin{proposition}[Bernstein-type inequality for $\psi_{1/2}$ RVs] \label{prop: Bernstein for psi_1/2}
	Let $X_1,\ldots,X_m$ be an independent, centered $\psi_{1/2}$ random variables. There is an absolute contant $C$ such that $S_m := \frac{1}{\sqrt{m}} \sum_{i=1}^m X_i$ is a $\psi_{1/2}$ random variable with $\psi_{1/2}$ norm satisfying
	\[
	\norm{S_m}_{\psi_{1/2}} \leq C\max_{1 \leq i \leq m}\norm{X_i}_{\psi_{1/2}}. 
	\]
	In particular, if $\max_{1 \leq i \leq m}\norm{X_i}_{\psi_{1/2}}$ is bounded above by a constant, for every $t \geq 0$, we have
	\[
	\P\braces*{\abs{S_m} \geq t} \leq 2\exp(-\sqrt{t/C}).
	\]
\end{proposition}

\begin{proof}
	This follows more or less immediately from the last two propositions. First, notice that
	\[
	\E\abs*{\sum_{i=1}^m X_i} \leq \paren*{\E{\abs*{\sum_{i=1}^m X_i}}^2}^{1/2} \leq C\sqrt{m} \max_{1 \leq i \leq m} \norm{X_i}_{\psi_{1/2}}
	\]
	Here, the first inequality is an application of Jensen's inequality, and the second uses the moment bound in Proposition \ref{prop: characterization of psi_alpha RVs}. Next, we compute $\psi_\alpha^{-1}(m) = (\log(m+1))^2$, and use Proposition \ref{prop: bound on psi_alpha norm of max}, we get
	\[
	\norm*{\max_{1 \leq i \leq m}\abs{X_i}}_{\psi_{\alpha}} \leq C (\log m)^2 \max_{1 \leq i \leq m} \norm{X_i}_{\psi_{\alpha}}.
	\]
	
	Finally, plug these two bounds into the inequality given by Proposition \ref{prop: bound on psi_alpha norm of sum}, and note that $\log(m+1)/\sqrt{m} \leq 5$. This completes the proof of the first statement. The tail bound follows from Proposition \ref{prop: characterization of psi_alpha RVs}.
\end{proof}

\section{Recovery using general geometric signal constraints} \label{sec: general geometric signal constraints}

The goal of this section is to prove Theorem \ref{thm: general recovery}, and to collate the necessary theoretical apparatus for doing so. First, we state the algorithm we propose for estimating $\boldx_*$ given general geometric constraints. We call this algorithm $\mathcal{K}$-PCA.

\begin{algorithm}[H]
	\caption{{\sc $\mathcal{K}$-PCA for MPR}}
	\begin{algorithmic}[1]              
		\REQUIRE Measurements $y_1,\ldots,y_m$, sampling vectors $\bolda_1,\ldots,\bolda_m$.
		\ENSURE An estimate $\hat{\boldx}$ for $\boldx_*$.
		\STATE Compute $\hat{\boldSigma}$ as defined in \eqref{eq: hatSigma definition}.
		\STATE Let $\hat{\boldX}$ be the solution to
		\begin{equation} \label{eq: $C$-PCA SDP}
		\max_{X \succeq 0} ~\inprod{\boldX,\hat{\boldSigma}} \quad \textnormal{subject to} \quad ~\boldX \in \mathcal{K}.
		\end{equation}
		\STATE Let $\hat{\boldx}$ be the leading eigenvector to $\hat{\boldX}$.
	\end{algorithmic}
	\label{alg: K-PCA algorithm}
\end{algorithm}

This can be seen as a tensorized version of the $1$-bit sensing algorithm proposed in \cite{Plan2013}. Our analysis will be also be similar to that in \cite{Plan2013}, but we will require a more general concentration result (Lemma \ref{lem: uniform deviaton on K}) that is derived via chaining. We introduce the requisite definitions as follows.

Let $(\mathcal{T},d)$ be a metric space. A sequence $\mathfrak{T} = (T_k)_{k \in \Z_+}$ of subsets of $\mathcal{T}$ is called \emph{admissible} if $\abs{\mathcal{T}_0} = 1$, and $\abs{\mathcal{T}_k} \leq 2^{2^k}$ for all $k \geq 1$. For any $0 < \alpha < \infty$, we define the \emph{$\gamma_\alpha$ functional} of $(\mathcal{T},d)$ to be
\begin{equation}
\gamma_\alpha(\mathcal{T},d) := \inf_{\mathcal{T}}\sup_{t \in T} \sum_{k=0}^\infty 2^{k/\alpha} d(t,\mathcal{T}_k).
\end{equation}
When $(\mathcal{T},d)$ is a subset of $\R^n$ with the Euclidean metric, Talagrand's comparison theory tells us that $\gamma_2(\mathcal{T},d) \asymp w(\mathcal{T})$, where $w(\mathcal{T})$ is the Gaussian width of $\mathcal{T}$ (see \cite{Ledoux1991,Vershynin}). Also note that for any positive constant $c > 0$, we have $\gamma_\alpha(T,cd) = c \gamma_\alpha(\mathcal{T},d)$. We shall use this fact later.

Now, let $d_1$ and $d_2$ be two metrics on $\mathcal{T}$. We say that a process $(Y_t)$ has \emph{mixed tail increments} with respect to $(d_1,d_2)$ if there are constants $c$ and $C$ such that for all $s, t \in \mathcal{T}$, we have the bound
\begin{equation} \label{def: mixed tail increments}
\P\paren{\abs{Y_s-Y_t} \geq c(\sqrt{u}d_2(s,t) + ud_1(s,t)) } \leq Ce^{-u}.
\end{equation}

\begin{lemma}[Suprema of mixed tail processes, Theorem 5, \cite{Dirksen2015}] \label{lem: sup bound for mixed tail process}
	If $(Y_t)_{t \in \mathcal{T}}$ has mixed tail increments, then there is a constant $C$ such that for any $u \geq 1$, with probability at least $1 - e^{-u}$,
	\begin{equation}
	\sup_{t \in \mathcal{T}}\abs{Y_t - Y_{t_0}} \leq C\paren{\gamma_2(\mathcal{T},d_2) + \gamma_1(\mathcal{T},d_1) + \sqrt{u}\diam(\mathcal{T},d_2) + u\diam(\mathcal{T},d_1)}.
	\end{equation}
\end{lemma}

Let us return to our situation, where we have a subset $\mathcal{K}$ of unit trace PSD matrices in $\R^{n\cross n}$. Fix $\boldx_*\boldx_*^T$, and real numbers $z_1,\ldots,z_m$. We define a process on the set $\mathcal{K}$ as follows. Let $\bolda_1,\ldots,\bolda_m$ be standard Gaussians. For each $\boldX \in \mathcal{K}$, we set
\[
Y_{\boldX} = \inprod{\textstyle{\sum_{i=1}^m} z_i (\bolda_i\bolda_i^T-\boldI_n), \boldX -\boldx_*\boldx_*^T}.
\]

We claim the following.

\begin{lemma}[Process increments] \label{lem: process has mixed tails}
	The process $Y_{\boldX}$ has mixed tail increments with respect to $(d_1,d_2)$, where $d_2(\boldX,\boldX') = (\sum_{i=1}^m z_i^2)^{1/2}\norm{\boldX-\boldX'}_F$, and $d_1(\boldX,\boldX') = \max_i \abs{z_i}\cdot \norm{\boldX-\boldX'}$.
\end{lemma}

\begin{proof}
	Fix $\boldX, \boldX' \in \mathcal{K}$, and for convenience, denote $\boldH = \boldX - \boldX'$. Then
	\begin{align*}
	Y_{\boldX} - Y_{\boldX'} & = \sum_{i=1}^m z_i\paren*{\bolda_i^T\boldH\bolda_i -\E\braces{\bolda_i^T\boldH\bolda_i}} \\
	& = \sum_{i=1}^mz_i\sum_{j=1}^n \lambda_j \paren*{\inprod{\bolda_i,\boldv_j}^2 -\E\braces{\inprod{\bolda_i,\boldv_j}^2}} \\
	& = \sum_{i=1}^m\sum_{j=1}^n z_i\lambda_j \paren*{\inprod{\bolda_i,\boldv_j}^2 -1},
	\end{align*}
	where $\boldH = \sum_{i=1}^m \lambda_i \boldv_i\boldv_i^T$ is the eigendecomposition of $\boldH$. Next, observe that by the independence of orthogonal Gaussian marginals, $\braces{\paren{\inprod{\bolda_i,\boldv_j}^2 -1} ~\colon 1 \leq i \leq m, ~1 \leq j \leq n}$ are independent, centered subexponential random variables with bounded subexponential norm. We may thus apply Bernstein's inequality to get
	\[
	\P\braces{\abs{Y_{\boldX} - Y_{\boldX'}} \geq t} \leq 2\exp\paren*{-c\min\braces*{\frac{t^2}{\sum_{i=1}^m\sum_{j=1}^n z_i^2\lambda_j^2},\frac{t}{\max_{i,j}\abs{z_i\lambda_j}}}}.
	\]
	Finally, observe that $\sum_{j=1}^n z_j^2 = \norm{\boldH}_F^2$ and $\max_{1 \leq j \leq n} \abs{z_j} = \norm{\boldH}$. One can now check that \eqref{def: mixed tail increments} is satisfied with respect to our chosen $d_1$ and $d_2$.
\end{proof}

\begin{lemma}[Uniform deviation bound] \label{lem: uniform deviaton on K}
	Let $\bolda_1,\ldots,\bolda_m$ be independent standard Gaussians, and suppose that Assumption \eqref{ass: Gaussian assumption} holds. Let $\mathcal{K}$ be a convex subset of the space of unit trace PSD matrices in $\R^{n \cross n}$. For any $\epsilon,\delta > 0$, if $m$ satisfies the lower bound \eqref{eq: sample complexity for general constraints}, then with probability at least $1-\delta$, we have
	\[
	\sup_{\boldX \in \mathcal{K}} \abs*{\inprod{\hat{\boldSigma}-\boldSigma,\boldX-\boldx_*\boldx_*^T}} \leq \epsilon^2,
	\]
\end{lemma}

\begin{proof}
	The proof of this concentration bound follows the same strategy as that in \cite{Plan2016a}. A priori, the process we are trying to control has heavy tails. To overcome this, we will use a decoupling argument together with conditioning.
	
	For each $\boldX \in \mathcal{K}$, denote $\boldH = \boldX-\boldx_*\boldx_*^T$ for convenience. Let $\boldP_{\boldx_*}$ and $\boldP_{\boldx_*^\perp}$ denote projection onto $\boldx_*$ and its orthogonal complement respectively. We can then write
	\begin{align} \label{eq: process decomposition}
	\inprod{\hat{\boldSigma}-\boldSigma,\boldH} & = \inprod{\boldP_{\boldx_*}(\hat{\boldSigma}-\boldSigma)\boldP_{\boldx_*},\boldH} + 2\inprod{\boldP_{\boldx_*^\perp}(\hat{\boldSigma}-\boldSigma)\boldP_{\boldx_*},\boldH} + \inprod{\boldP_{\boldx_*^\perp}(\hat{\boldSigma}-\boldSigma)\boldP_{\boldx_*^\perp},\boldH} \nonumber \\
	& = \inprod{\boldP_{\boldx_*}(\hat{\boldSigma}-\boldSigma)\boldP_{\boldx_*},\boldH} + 2\inprod{\boldP_{\boldx_*^\perp}\hat{\boldSigma}\boldP_{\boldx_*},\boldH} + \inprod{\boldP_{\boldx_*^\perp}\hat{\boldSigma}\boldP_{\boldx_*^\perp},\boldH}.
	\end{align}
	We shall bound the three terms on the right separately.
	
	Recalling that $\boldP_{\boldx_*} = \boldx_*\boldx_*^T$, we see that the first term can be written as
	\begin{align*}
	\inprod{\boldP_{\boldx_*}(\hat{\boldSigma}-\boldSigma)\boldP_{\boldx_*},\boldH} & = \boldx_*^T(\hat{\boldSigma}-\boldSigma)\boldx_*  \boldx_*^T\boldH\boldx_* \\
	&= \sqbracket*{\frac{1}{m}\sum_{i=1}^m y_i (\inprod{\bolda_i,\boldx_*}^2-1) - \E \braces{y(\inprod{\bolda_i,\boldx_*}^2-1)}}\boldx_*^T\boldH\boldx_*.
	\end{align*}
	Notice that $\boldx_*^T\boldH\boldx_* \leq 1$. Meanwhile, the term in the square brackets is the average of independent, centered, $\psi_{1/2}$ random variables. Using Proposition \ref{prop: Bernstein for psi_1/2}, we have a probability at least $1-\delta/4$ event over which the following bound holds:
	\begin{equation} \label{eq: bound for first term}
	\sup_{\boldX \in \mathcal{K}} \abs*{\inprod{\boldP_{\boldx_*}(\hat{\boldSigma}-\boldSigma)\boldP_{\boldx_*},\boldH}} \leq \frac{C\log(1/\delta)^2}{\sqrt{m}}.
	\end{equation}
	
	For the third term in \eqref{eq: process decomposition}, we write
	\begin{align*}
	\boldP_{\boldx_*^\perp}\hat{\boldSigma}\boldP_{\boldx_*^\perp} = \frac{1}{m}\sum_{i=1}^my_i\paren*{(\boldP_{\boldx_*^\perp}\bolda_i)(\boldP_{\boldx_*^\perp}\bolda_i)^T - \boldP_{\boldx_*^\perp}}.
	\end{align*}
	Since $\boldy_i$ and $\boldP_{\boldx_*^\perp}\bolda_i$ are independent, we may \textit{decouple} them. In other words, we replace each $\bolda_i$ with a fully independent copy $\tilde{\bolda}_i$. We can therefore write
	\begin{align*}
	\inprod{\boldP_{\boldx_*^\perp}\hat{\boldSigma}\boldP_{\boldx_*^\perp},\boldH} & = \inprod*{\textstyle{\frac{1}{m}\sum_{i=1}^my_i\paren*{(\boldP_{\boldx_*^\perp}\tilde{\bolda}_i)(\boldP_{\boldx_*^\perp}\tilde{\bolda}_i)^T - \boldP_{\boldx_*^\perp}},\boldH}} \\
	& = \inprod*{\textstyle{\frac{1}{m}\sum_{i=1}^my_i\paren{\tilde{\bolda_i}\tilde{\bolda_i}^T - \boldI_n},\boldP_{\boldx_*^\perp}\boldH\boldP_{\boldx_*^\perp}}}
	\end{align*}
	
	Fix the randomness with respect to $\bolda_1,\ldots,\bolda_m$, $y_1,\ldots,y_m$, conditioning on the probability $1-\delta/4$ event that the three statements in Lemma \ref{lem: helper lemma for uniform deviation} hold. With respect to the $\tilde{\bolda}_i$'s, Lemma \ref{lem: process has mixed tails} shows us that this process indexed over $\boldX \in \mathcal{K}$ has mixed tails. Furthermore, since $\mathcal{K}$ is a subset of the nuclear norm ball, its diameter with respect to both the Frobenius and operator norms is bounded by 2. Lemma \ref{lem: sup bound for mixed tail process} then gives us another probability $1-\delta/4$ event over which
	\begin{align} \label{eq: bound for third term}
	\sup_{\boldX \in \mathcal{K}}\abs*{\inprod{\boldP_{\boldx_*^\perp}\hat{\boldSigma}\boldP_{\boldx_*^\perp},\boldH}} & \leq \frac{C}{m} \paren*{\paren{\sum_{i=1}^m y_i^2}^{1/2} \paren*{\gamma_2(\mathcal{K},\norm{\cdot}_2)+\sqrt{\log(1/\delta)}} + \max_{1 \leq i \leq m} \abs{y_i} \paren*{\gamma_1(\mathcal{K},\norm{\cdot}) +\log(1/\delta)}} \nonumber \\
	& \leq C\paren*{\frac{\gamma_2(\mathcal{K},\norm{\cdot}_2)+\sqrt{\log(1/\delta)}}{\sqrt{m}}+ \frac{\log m (\gamma_1(\mathcal{K},\norm{\cdot})+\log(1/\delta))}{m}}.
	\end{align}
	
	Finally, for the second term in \eqref{eq: process decomposition}, we again decouple, writing
	\begin{align} \label{eq: intermediate quantity in bound for second term}
	\inprod{\boldP_{\boldx_*^\perp}\hat{\boldSigma}\boldP_{\boldx_*},\boldH} & = \inprod*{\textstyle{\frac{1}{m}\sum_{i=1}^m} y_i \inprod{\bolda_i,\boldx_*}\boldx_*(\boldP_{\boldx_*^\perp}\bolda_i)^T,\boldH} \nonumber \\
	& = \inprod*{\textstyle{\frac{1}{m}\sum_{i=1}^m} y_i \inprod{\bolda_i,\boldx_*}\boldx_*(\boldP_{\boldx_*^\perp}\tilde{\bolda}_i)^T,\boldH} \nonumber\\
	& = \inprod*{\textstyle{\frac{1}{m}\sum_{i=1}^m} y_i \inprod{\bolda_i,\boldx_*}\tilde{\bolda}_i,\boldP_{\boldx_*^\perp}\boldH\boldx_* }.
	\end{align}
	
	Again, fix the randomness with respect to $\bolda_1,\ldots,\bolda_m$, $y_1,\ldots,y_m$, and remember that we have conditioned on the event that the three statements in Lemma \ref{lem: helper lemma for uniform deviation} hold. Then with respect to the $\tilde{\bolda}_i$'s, the quantity on the right in \eqref{eq: intermediate quantity in bound for second term} is a centered Gaussian random variable with variance
	\begin{align*}
	\sigma^2 = \frac{1}{m^2}\sum_{i=1}^m y_i^2\inprod{\bolda_i,\boldx_*}^2 \norm{\boldP_{\boldx_*^\perp}\boldH\boldx_*}_2^2 \leq \frac{C\norm{\boldH}^2}{m}.
	\end{align*}
	
	Therefore $\paren{\inprod{\boldP_{\boldx_*^\perp}\hat{\boldSigma}\boldP_{\boldx_*},\boldH}}_\boldX$ is a process indexed by $\boldX$ which has subgaussian increments with respect to the operator norm. We use an analogue of Theorem \ref{lem: sup bound for mixed tail process} (see Theorem 3.2 in \cite{Dirksen2015}) to obtain a probability $1-\delta/4$ event over which
	\begin{align} \label{eq: bound for second term}
	\sup_{\boldX \in \mathcal{K}}\abs*{\inprod{\boldP_{\boldx_*^\perp}\hat{\boldSigma}\boldP_{\boldx_*},\boldH}} \leq C\paren*{\frac{\gamma_2(\mathcal{K},\norm{\cdot}_2)+\sqrt{\log(1/\delta)}}{\sqrt{m}}}.
	\end{align}
	
	Combining the bounds \eqref{eq: bound for first term}, \eqref{eq: bound for second term}, and \eqref{eq: bound for third term} gives us the statement we want.
\end{proof}

\begin{lemma} \label{lem: helper lemma for uniform deviation}
	Let $\bolda_1,\ldots,\bolda_m$ be independent standard Gaussians, and suppose that Assumption \eqref{ass: Gaussian assumption} holds. Then for any $\delta > 0$, so long as $m \geq C/\delta$, the following three statements hold simultaneously with probability at least $1- \delta/4$.
	\begin{enumerate}
		\item $\displaystyle{\sum_{i=1}^m f(\inprod{\bolda_i,\boldx_*})^2 \leq Cm}$.
		\item $\displaystyle{\sum_{i=1}^m y_i^2 f(\inprod{\bolda_i,\boldx_*})^2 \leq Cm}$.
		\item $\displaystyle{\max_{1 \leq i \leq m}f(\inprod{\bolda_i,\boldx_*}) \leq C\log m }$.
	\end{enumerate}
\end{lemma}

\begin{proof}
	Exactly the same as in Lemma \ref{lem: helper lemma for matrix concentration}.
\end{proof}

\begin{proof}[Proof of Theorem \ref{thm: general recovery}]
	We repeat the argument of Theorem \ref{thm: sparse recovery for Gaussian}, but replace the H{\"o}lder's inequality bound of $\inprod{\hat{\boldSigma}-\boldSigma,\boldX-\boldx_*\boldx_*^T}$ therein with the uniform deviation bound supplied by Lemma \ref{lem: uniform deviaton on K}.
\end{proof}

\end{document}